\documentclass[12pt,abstract]{scrartcl}
\usepackage{graphicx}
\usepackage{amsmath,amsthm}
\makeatletter
\renewcommand{\theenumii}{\@roman\c@enumii}
\makeatother

\usepackage{natbib}
\newtheorem{lemma}{Lemma}
\newtheorem{theorem}{Theorem}

\usepackage{newtxtext}
\usepackage[varg,vvarbb]{newtxmath}

\begin{document}

\title{Sion's mini-max theorem and Nash equilibrium in a five-players game with two groups which is zero-sum and symmetric in each group\thanks{This work was supported by Japan Society for the Promotion of Science KAKENHI Grant Number 15K03481 and 18K0159.}}

\author{
Atsuhiro Satoh\thanks{atsatoh@hgu.jp}\\[.01cm]
Faculty of Economics, Hokkai-Gakuen University,\\[.02cm]
Toyohira-ku, Sapporo, Hokkaido, 062-8605, Japan,\\[.01cm]
\textrm{and} \\[.1cm]
Yasuhito Tanaka\thanks{yasuhito@mail.doshisha.ac.jp}\\[.01cm]
Faculty of Economics, Doshisha University,\\
Kamigyo-ku, Kyoto, 602-8580, Japan.\\}

\date{}

\maketitle
\thispagestyle{empty}

\vspace{-1.2cm}

\begin{abstract}
We consider the relation between Sion's minimax theorem for a continuous function and a Nash equilibrium in a five-players game with two groups which is zero-sum and symmetric in each group. We will show the following results.
\begin{enumerate}
	\item The existence of Nash equilibrium which is symmetric in each group implies Sion's minimax theorem for a pair of playes in each group. 
	\item Sion's minimax theorem for a pair of playes in each group imply the existence of a Nash equilibrium which is symmetric in each group.
\end{enumerate}
Thus, they are equivalent. An example of such a game is a relative profit maximization game in each group under oligopoly with two groups such that firms in each group have the same cost functions and maximize their relative profits in each group, and the demand functions are symmetric for the firms in each group.
\end{abstract}

\begin{description}
	\item[Keywords:] five-players zero-sum game, two groups, Nash equilibrium, Sion's minimax theorem
\end{description}

\begin{description}
	\item[JEL Classification:] C72
\end{description}

\section{Introduction}

We consider the relation between Sion's minimax theorem for a continuous function and existence of a Nash equilibrium in a five-players game with two groups which is zero-sum and symmetric in each group. There are five players, A, B, C, D and E. Players A, B and E are in one group, and Players C and D are in the other group. Players A, B and E have the same payoff functions and strategy spaces, and they play a game which is zero-sum in this group, that is, the sum of the payoffs of Players A, B and E is zero. Similarly, Players C and D have the same payoff functions and strategy spaces, and they play a game which is zero-sum in this group, that is, the sum of the payoffs of Players C and D is zero.

We will show the following results.
\begin{enumerate}
	\item The existence of Nash equilibrium which is symmetric in each group implies Sion's minimax theorem for a pair of playes in each group. 
	\item Sion's minimax theorem for a pair of playes in each group imply the existence of a Nash equilibrium which is symmetric in each group.
\end{enumerate}
Thus, they are equivalent. The coincidence of the maximin strategy and the minimax strategy is assumed in Assumption \ref{as1}

An example of such a game is a relative profit maximization game in each group under oligopoly with two groups such that firms in each group have the same cost functions and maximize their relative profits in each group, and the demand functions are symmetric for the firms in each group. Consider a five firms oligopoly. Let $\bar{\pi}_A$, $\bar{\pi}_B$, $\bar{\pi}_C$, $\bar{\pi}_D$ and $\bar{\pi}_E$ be the absolute profits of, respectively, Firms A, B, C, D and E. Then, the relative profits of Firms A, B and E are
\[\pi_A=\bar{\pi}_A-\frac{1}{2}(\bar{\pi}_B+\bar{\pi}_E),\]
\[\pi_B=\bar{\pi}_B-\frac{1}{2}(\bar{\pi}_A+\bar{\pi}_E),\]
\[\pi_E=\bar{\pi}_B-\frac{1}{2}(\bar{\pi}_A+\bar{\pi}_B).\]
The relative profits of Firms C and D are
\[\pi_C=\bar{\pi}_C-\bar{\pi}_D,\]
\[\pi_D=\bar{\pi}_D-\bar{\pi}_C.\]
We see
\[\pi_A+\pi_B+\pi_E=0,\]
\[\pi_C+\pi_D=0.\]
Firms A, B, C, D and E maximize, respectively, $\pi_A$, $\pi_B$, $\pi_C$, $\pi_D$ and $\pi_E$. Thus, the relative profit maximization game in each group is a zero-sum game\footnote{About relative profit maximization under imperfect competition please see \cite{mm}, \cite{ebl2}, \cite{eb2}, \cite{st}, \cite{eb1}, \cite{ebl1} and \cite{redondo}}. In Section \ref{ex} we present an example of relative profit maximization in each group under oligopoly with two groups. In that example we assume that the demand functions for Firms A, B and E are symmetric, those for Firm C and are symmetric, Firms A, B and E have the same cost function, and Firms C and D have the same cost function.

\section{The model and Sion's minimax theorem}

Consider a five-players game with two groups. The game is zero-sum in each group. There are five players, A, B, C, D and E. The strategic variables for Players A, B, C, D and E are, respectively, $s_A$, $s_B$, $s_C$, $s_D$, $s_E$, and ($s_A, s_B, s_C, s_D, s_E)\in S_A\times S_B\times S_C\times S_D\times S_E$. $S_A$, $S_B$, $S_C$, $S_D$, $S_E$ are convex and compact sets in linear topological spaces. The payoff function of each player is $u_i(s_A, s_B, s_C, s_D, s_E),\ i=A, B, C, D, E$. They are real valued functions on $S_A\times S_B\times S_C\times S_D\times S_E$. We assume 
\begin{quote}
$u_A$, $u_B$, $u_C$, $u_D$ and $u_E$ are continuous on $S_A\times S_B \times S_C \times S_D \times S_E$, quasi-concave on $S_i$ for each $s_j\in S_j,\ j\neq i$, and quasi-convex on $S_j$ for $j\neq i$ for each $s_i\in S_i,\ i=A, B, C, D, E$,
\end{quote}
and
\[S_A=S_B=S_E,\ \mathrm{and}\ S_C=S_D.\]
There are two groups. Players A,  B and E are in one group, and Players C and D are in the other group. The game is symmetric for the players in each group in the sense that Players A, B and E have the same payoff functions and strategy spaces, and Players C and D have the same payoff functions and strategy spaces. 

The game is zero-sum in each group. Therefore,
\begin{equation}
u_A(s_A, s_B, s_C, s_D, s_E)+u_B(s_A, s_B, s_C, s_D)+u_E(s_A, s_B, s_C, s_D, s_E)=0,\label{e1}
\end{equation}
\begin{equation}
u_C(s_A, s_B, s_C, s_D, s_E)+u_D(s_A, s_B, s_C, s_D, s_E)=0,\label{e1a}
\end{equation}
for given $(s_A, s_B, s_C, s_D,, s_E)$.

Sion's minimax theorem (\cite{sion}, \cite{komiya}, \cite{kind}) for a continuous function is stated as follows.
\begin{lemma}
Let $X$ and $Y$ be non-void convex and compact subsets of two linear topological spaces, and let $f:X\times Y \rightarrow \mathbb{R}$ be a function that is continuous and quasi-concave in the first variable and continuous and quasi-convex in the second variable. Then
\[\max_{x\in X}\min_{y\in Y}f(x,y)=\min_{y\in Y}\max_{x\in X}f(x,y).\] \label{l1}
\end{lemma}
We follow the description of this theorem in \cite{kind}.

Let $s_C$, $s_D$ and $s_E$ be given. Then, $u_A(s_A, s_B, s_C, s_D, s_E)$ is a function of $s_A$ and $s_B$. We can apply Lemma \ref{l1} to such a situation, and get the following equation.
\begin{equation}
\max_{s_A\in S_A}\min_{s_B\in S_B}u_A(s_A,s_B,s_C,s_D,s_E)=\min_{s_B\in S_B}\max_{s_A\in S_A}u_A(s_A, s_B, s_C,s_D,s_E).\label{as0}
\end{equation}
By symmetry we require
\begin{equation*}
\max_{s_E\in S_E}\min_{s_B\in S_B}u_E(s_B,s_B,s_C,s_D,s_E)=\min_{s_B\in S_B}\max_{s_E\in S_E}u_E(s_B, s_B, s_C,s_D,s_E).
\end{equation*}
However, we do not require
\begin{equation*}
\max_{s_B\in S_B}\min_{s_A\in S_A}u_B(s_A,s_B,s_C,s_D,s_E)=\min_{s_A\in S_A}\max_{s_B\in S_B}u_B(s_A, s_B, s_C,s_D,s_E),
\end{equation*}
nor
\begin{equation*}
\max_{s_A\in S_A}\min_{s_E\in S_E}u_A(s_A,s_B,s_C,s_D,s_E)=\min_{s_E\in S_E}\max_{s_A\in S_A}u_A(s_A, s_B, s_C,s_D,s_E),
\end{equation*}
nor
\begin{equation*}
\max_{s_E\in S_E}\min_{s_A\in S_A}u_E(s_A,s_B,s_C,s_D,s_E)=\min_{s_A\in S_A}\max_{s_E\in S_E}u_E(s_A, s_B, s_C,s_D,s_E),
\end{equation*}
nor
\begin{equation*}
\max_{s_B\in S_B}\min_{s_E\in S_E}u_B(s_B,s_B,s_C,s_D,s_E)=\min_{s_E\in S_E}\max_{s_B\in S_B}u_B(s_B, s_B, s_C,s_D,s_E).
\end{equation*}

Similarly, given $s_A$, $s_B$ and $s_E$ we get
\begin{equation}
\max_{s_C\in S_C}\min_{s_D\in S_D}u_C(s_A,s_B,s_C,s_D,s_E)=\min_{s_D\in S_D}\max_{s_C\in S_C}u_C(s_A, s_B, s_C,s_D,s_E).\label{as0a}
\end{equation}
We do not require
\begin{equation*}
\max_{s_D\in S_D}\min_{s_C\in S_C}u_D(s_A,s_B,s_C,s_D,s_E)=\min_{s_C\in S_C}\max_{s_D\in S_D}u_D(s_A, s_B, s_C,s_D,s_E).
\end{equation*}
We assume that $\arg\max_{s_A\in S_A}\min_{s_B\in S_B}u_A(s_A, s_B, s_C, s_D,s_E)$, $\arg\max_{s_E\in S_E}\min_{s_B\in S_B}u_E(s_A, s_B, s_C, s_D,s_E)$ and $\arg\max_{s_C\in S_C}\min_{s_D\in S_D}u_C(s_A,s_B,s_C,s_D,s_E)$ are unique, that is, single-valued. By the maximum theorem they are continuous. Also, throughout this paper we assume that the maximin strategy and the minimax strategy of players in any situation are unique, and the best responses of players in any situation are unique.

Let us consider a point such that $s_A=s_B=s_E=s$ and $s_C=s_D=s'$, and consider the following function.
\[
\begin{pmatrix}
s\\
s'
\end{pmatrix}
\rightarrow 
\begin{pmatrix}
\arg\max_{s_A\in S_A}\min_{s_B\in S_B}u_A(s_A,s_B,s',s',s)\\\arg\max_{s_C\in S_C}\min_{s_D\in S_D}u_C(s,s,s_C,s_D,s)
\end{pmatrix}.
\]
Since $u_A$ and $u_C$ are continuous, $S_A=S_B=S_E$ are compact and $S_C=S_D$ are compact, these functions are also continuous. Thus, there exists a fixed point of $(s,s')$. Denote it by $(\tilde{s},\hat{s})$. It satisfies
\begin{equation}
\tilde{s}=\arg\max_{s_A\in S_A}\min_{s_B\in S_B}u_A(s_A,s_B,\hat{s},\hat{s},\tilde{s}),\label{fix}
\end{equation}
\begin{equation}
\hat{s}=\arg\max_{s_C\in S_C}\min_{s_D\in S_D}u_C(\tilde{s},\tilde{s},s_C,s_D,\tilde{s}).\label{fixa}
\end{equation}
By symmetry we get
\[\tilde{s}=\arg\max_{s_E\in S_E}\min_{s_B\in S_B}u_E(\tilde{s},s_B,\hat{s},\hat{s},s_E),\]
However we do not require 
\[\tilde{s}=\arg\max_{s_B\in S_B}\min_{s_A\in S_A}u_B(s_A,s_B,\hat{s},\hat{s},\tilde{s}),\]
\[\tilde{s}=\arg\max_{s_A\in S_A}\min_{s_E\in S_E}u_A(s_A,\tilde{s},\hat{s},\hat{s},s_E),\]
\[\tilde{s}=\arg\max_{s_E\in S_E}\min_{s_A\in S_A}u_E(s_A,\tilde{s},\hat{s},\hat{s},s_E),\]
\[\tilde{s}=\arg\max_{s_B\in S_B}\min_{s_E\in S_E}u_B(\tilde{s},s_B,\hat{s},\hat{s},s_E),\]
\[\hat{s}=\arg\max_{s_D\in S_D}\min_{s_C\in S_C}u_D(\tilde{s},\tilde{s},s_C,s_D,\tilde{s}).\]

\section{The main results}

Consider a Nash equilibrium which is symmetric in each group. Let $s_A^*$, $s_B^*$, $s_C^*$, $s_D^*$, $s_E^*$ be the values of $s_A$, $s_B$, $s_C$, $s_D$, $s_E$ which, respectively, maximize $u_A$, $u_B$, $u_C$, $u_D$, $u_E$, that is, 
\[u_A(s_A^*,s_B^*,s_C^*,s_D^*,s_E^*)\geq u_A(s_A,s_B^*,s_C^*,s_D^*,s_E^*)\ \mathrm{for\ any}\ s_A\in S_A,\]
\[u_B(s_A^*,s_B^*,s_C^*,s_D^*,s_E^*)\geq u_B(s_A^*,s_B,s_C^*,s_D^*,s_E^*)\ \mathrm{for\ any}\ s_B\in S_B,\]
\[u_C(s_A^*,s_B^*,s_C^*,s_D^*,s_E^*)\geq u_C(s_A^*,s_B^*,s_C,s_D^*,s_E^*)\ \mathrm{for\ any}\ s_C\in S_C,\]
\[u_D(s_A^*,s_B^*,s_C^*,s_D^*,s_E^*)\geq u_D(s_A^*,s_B^*,s_C^*,s_D,s_E^*)\ \mathrm{for\ any}\ s_D\in S_D,\]
and
\[u_E(s_A^*,s_B^*,s_C^*,s_D^*,s_E^*)\geq u_E(s_A^*,s_B^*,s_C^*,s_D^*,s_E)\ \mathrm{for\ any}\ s_E\in S_E.\]

If the Nash equilibrium is symmetric in each group, $s_A^*$, $s_B^*$ and $s_E^*$ are equal, and $s_C^*$ and $s_D^*$ are equal.

We show the  following theorem.
\begin{theorem}
The existence of Nash equilibrium which is symmetric in each group implies Sion's minimax theorem with the coincidence of the maximin strategy and the minimax strategy. 
\label{t1}
\end{theorem}

\begin{proof}
Let $(s_A^*,s_B^*,s_C^*,s_D^*,s_E^*)$ be a Nash equilibrium which is symmetric in each group. Denote $s_A^*=s_B^*=s_E^*=s^*$, $s_C^*=s_D^*=s^{**}$. Since the game is zero-sum in each group,
\[u_A(s_A,s^*,s^{**},s^{**},s^*)+u_B(s_A,s^*,s^{**},s^{**},s^*)+u_E(s_A,s^*,s^{**},s^{**},s^*)=0,\]
and
\[u_C(s^*,s^*,s_C,s^{**},s^*)+u_D(s^*,s^*,s_C,s^{**},s^*)=0,\]
imply
\[u_A(s_A,s^*,s^{**},s^{**},s^*)=-(u_B(s_A,s^*,s^{**},s^{**},s^*)+u_E(s_A,s^*,s^{**},s^{**},s^*)),\]
\[u_C(s^*,s^*,s_C,s^{**},s^*)=-u_D(s^*,s^*,s_C,s^{**},s^*)\]
By symmetry for Players A, B and E
\[u_A(s_A,s^*,s^{**},s^{**},s^*)=-2u_B(s_A,s^*,s^{**},s^{**},s^*).\]
These equations hold for any $s_A$ and $s_C$. Therefore,
\[\arg\max_{s_A\in S_A}u_A(s_A,s^*,s^{**},s^{**},s^*)=\arg\min_{s_A\in S_A}u_B(s_A,s^*,s^{**},s^{**},s^*),\]
\[\arg\max_{s_C\in S_C}u_C(s^*,s^*,s_C,s^{**},s^*)=\arg\min_{s_C\in S_C}u_D(s^*,s^*,s_C,s^{**},s^*).\]
By the assumption of uniqueness of the best responses, they are unique. By symmetry for each group
\[\arg\max_{s_A\in S_A}u_A(s_A,s^*,s^{**},s^{**},s^*)=\arg\min_{s_B\in S_B}u_A(s^*,s_B,s^{**},s^{**},s^*),\]
\[\arg\max_{s_C\in S_C}u_C(s^*,s^*,s_C,s^{**},s^*)=\arg\min_{s_D\in S_D}u_C(s^*,s^*,s^{**},s_D,s^*).\]
Therefore,
\[u_A(s^*,s^*,s^{**},s^{**},s^*)=\min_{s_B\in S_B}u_A(s^*,s_B,s^{**},s^{**},s^*)\leq u_A(s^*,s_B,s^{**},s^{**},s^*),\]
\[u_C(s^*,s^*,s^{**},s^{**},s^*)=\min_{s_D\in S_D}u_C(s^*,s^*,s^{**},s_D,s^*)\leq u_C(s^*,s^*,s^{**},s_D,s^*).\]
We get
\[\max_{s_A\in S_A}u_A(s_A,s^*,s^{**},s^{**},s^*)=u_A(s^*,s^*,s^{**},s^{**},s^*)=\min_{s_B\in S_B}u_A(s^*,s_B,s^{**},s^{**},s^*),\]
\[\max_{s_C\in S_C}u_C(s^*,s^*,s_C,s^{**},s^*)=u_C(s^*,s^*,s^{**},s^{**},s^*)=\min_{s_D\in S_D}u_C(s^*,s^*,s^{**},s_D,s^*),\]
They mean
\begin{align}
&\min_{s_B\in S_B}\max_{s_A\in S_A}u_A(s_A,s_B,s^{**},s^{**},s^*)\leq \max_{s_A\in S_A}u_A(s_A,s^*,s^{**},s^{**},s^*) \label{e3} \\
=&\min_{s_B\in S_B}u_A(s^*,s_B,s^{**},s^{**},s^*)\leq \max_{s_A\in S_A}\min_{s_B\in S_B}u_A(s_A,s_B,s^{**},s^{**},s^*).\notag
\end{align}
and
\begin{align}
&\min_{s_D\in S_D}\max_{s_C\in S_C}u_C(s^*,s^*,s_C,s_D,s^*)\leq \max_{s_C\in S_C}u_C(s^*,s^*,s_C,s^{**},s^*) \label{e3a} \\
=&\min_{s_D\in S_D}u_C(s^*,s^*,s^{**},s_D,s^*)\leq \max_{s_C\in S_C}\min_{s_D\in S_D}u_C(s^*,s^*,s_C,s_D,s^*).\notag
\end{align}
On the other hand, since
\[\min_{s_B\in S_B}u_A(s_A,s_B,s^{**},s^{**},s^*)\leq u_A(s_A,s_B,s^{**},s^{**},s^*),\]
\[\min_{s_C\in S_C}u_C(s^*,s^*,s_C,s_D,s^*)\leq u_C(s^*,s^*,s_C,s_D,s^*),\]
we have
\[\max_{s_A\in S_A}\min_{s_B\in S_B}u_A(s_A,s_B,s^{**},s^{**},s^*)\leq \max_{s_A\in S_A}u_A(s_A,s_B,s^{**},s^{**},s^*),\]
\[\max_{s_C\in S_C}\min_{s_D\in S_D}u_C(s^*,s^*,s_C,s_D,s^*)\leq \max_{s_C\in S_C}u_C(s^*,s^*,s_C,s_D,s^*).\]
These inequalities hold for any $s_B$ and $s_D$. Thus,
\[\max_{s_A\in S_A}\min_{s_B\in S_B}u_A(s_A,s_B,s^{**},s^{**},s^*)\leq \min_{s_B\in S_B}\max_{s_A\in S_A}u_A(s_A,s_B,s^{**},s^{**},s^*),\]
\[\max_{s_C\in S_C}\min_{s_D\in S_D}u_C(s^*,s^*,s_C,s_D,s^*)\leq \min_{s_D\in S_D}\max_{s_C\in S_C}u_C(s^*,s^*,s_C,s_D,s^*),\]
With (\ref{e3}) and (\ref{e3a}), we obtain
\begin{equation}
\max_{s_A\in S_A}\min_{s_B\in S_B}u_A(s_A,s_B,s^{**},s^{**},s^*)=\min_{s_B\in S_B}\max_{s_A\in S_A}u_A(s_A,s_B,s^{**},s^{**},s^*),\label{t1-1}
\end{equation}
\begin{equation}
\max_{s_C\in S_C}\min_{s_D\in S_D}u_C(s^*,s^*,s_C,s_D,s^*)=\min_{s_D\in S_D}\max_{s_C\in S_C}u_C(s^*,s^*,s_C,s_D,s^*).\label{t1-1a}
\end{equation}
By symmetry for each group
\[\max_{s_A\in S_A}\min_{s_E\in S_E}u_A(s_A,s^*,s^{**},s^{**},s_E)=\min_{s_E\in S_E}\max_{s_A\in S_A}u_A(s_A,s^*,s^{**},s^{**},s_E),\]
\[\max_{s_D\in S_D}\min_{s_C\in S_C}u_D(s^*,s^*,s_C,s_D,s^*)=\min_{s_C\in S_C}\max_{s_D\in S_D}u_D(s^*,s^*,s_C,s_D,s^*),\]
and so on. (\ref{e3}), (\ref{e3a}), (\ref{t1-1}) and (\ref{t1-1a}) imply
\[\max_{s_A\in S_A}\min_{s_B\in S_B}u_A(s_A,s_B,s^{**},s^{**},s^*)=\max_{s_A\in S_A}u_A(s_A,s^*,s^{**},s^{**},s^*),\]
\[\max_{s_C\in S_C}\min_{s_D\in S_D}u_C(s^*,s^*,s_C,s_D,s^*)=\max_{s_C\in S_C}u_C(s^*,s^*,s_C,s^{**},s^*),\]
\[\min_{s_B\in S_B}\max_{s_A\in S_A}u_A(s_A,s_B,s^{**},s^{**},s^*)=\min_{s_B\in S_B}u_A(s^*,s_B,s^{**},s^{**},s^*),\]
\[\min_{s_D\in S_D}\max_{s_C\in S_C}u_C(s^*,s^*,s_C,s_D,s^*)=\min_{s_D\in S_D}u_C(s^*,s^*,s^{**},s_D,s^*).\]

From
\[\min_{s_B\in S_B}u_A(s_A,s_B,s^{**},s^{**},s^*)\leq u_A(s_A,s^*,s^{**},s^{**},s^*),\]
\[\min_{s_D\in S_D}u_C(s^*,s^*,s_C,s_D,s^*)\leq u_C(s^*,s^*,s_C,s^{**},s^*),\]
\[\max_{s_A\in S_A}\min_{s_B\in S_B}u_A(s_A,s_B,s^{**},s^{**},s^*)=\max_{s_A\in S_A}u_A(s_A,s^*,s^{**},s^{**}, s^*),\]
and
\[\max_{s_C\in S_C}\min_{s_D\in S_D}u_C(s^*,s^*,s_C,s_D,s^*)=\max_{s_C\in S_C}u_C(s^*,s^*,s_C,s^{**}, s^*),\]
we have
\[\arg\max_{s_A\in S_A}\min_{s_B\in S_B}u_A(s_A,s_B,s^{**},s^{**},s^*)=\arg\max_{s_A\in S_A}u_A(s_A,s^*,s^{**},s^{**},s^*)=s^*,\]
\[\arg\max_{s_C\in S_C}\min_{s_D\in S_D}u_C(s^*,s^*,s_C,s_D,s^*)=\arg\max_{s_C\in S_C}u_C(s^*,s^*s_C,s^*,s^*)=s^{**}.\]
From
\[\max_{s_A\in S_A}u_A(s_A,s_B,s^{**},s^{**},s^*)\geq u_A(s^*,s_B,s^{**},s^{**},s^*),\]
\[\max_{s_C\in S_C}u_C(s^*,s^*,s_C,s_D,s^*)\geq u_C(s^*,s^*,s^{**},s_D,s^*),\]
\[\min_{s_B\in S_B}\max_{s_A\in S_A}u_A(s_A,s_B,s^{**},s^{**},s^*)=\min_{s_B\in S_B}u_A(s^*,s_B,s^{**},s^{**},s^*),\]
and
\[\min_{s_D\in S_D}\max_{s_C\in S_C}u_C(s^*,s^*,s_C,s_D,s^*)=\min_{s_D\in S_D}u_C(s^*,s^*,s^{**},s_D,s^*),\]
we get
\[\arg\min_{s_B\in S_B}\max_{s_A\in S_A}u_A(s_A,s_B,s^{**},s^{**},s^*)=\arg\min_{s_B\in S_B}u_A(s^*,s_B,s^{**},s^{**},s^*)=s^*,\]
\[\arg\min_{s_D\in S_D}\max_{s_C\in S_C}u_C(s^*,s^*,s_C,s_D,s^*)=\arg\min_{s_D\in S_D}u_C(s^*,s^*,s^{**},s_D,s^*)=s^{**}.\]
Therefore,
\begin{equation}
\arg\max_{s_A\in S_A}\min_{s_B\in S_B}u_A(s_A,s_B,s^{**},s^{**},s^*)=\arg\min_{s_B\in S_B}\max_{s_A\in S_A}u_A(s_A,s_B,s^{**},s^{**},s^*)=s^*,\label{t1-2}
\end{equation}
\begin{equation}
\arg\max_{s_C\in S_C}\min_{s_D\in S_D}u_C(s^*,s^*,s_C,s_D,s^*)=\arg\min_{s_D\in S_D}\max_{s_C\in S_C}u_C(s^*,s^*,s_C,s_D,s^*)=s^{**}.\label{t1-2a}
\end{equation}
By symmetry for each group we get
\[\arg\max_{s_A\in S_A}\min_{s_E\in S_E}u_A(s_A,s^*,s^{**},s^{**},s_E)=\arg\min_{s_E\in S_E}\max_{s_A\in S_A}u_A(s_A,s^*,s^{**},s^{**},s_E)=s^*,\]
\[\arg\max_{s_D\in S_D}\min_{s_C\in S_C}u_D(s^*,s^*,s_C,s_D,s^*)=\arg\min_{s_C\in S_C}\max_{s_D\in S_D}u_D(s^*,s^*,s_C,s_D,s^*)=s^{**},\]
and so on.
\end{proof}

Next we show the following theorem. 
\begin{theorem}
Sion's minimax theorem with the coincidence of the maximin strategy and the minimax strategy imply the existence of a Nash equilibrium which is symmetric in each group.
\end{theorem}
\begin{proof}
Let $\tilde{s}$ and $\hat{s}$ be the values of $s_i,\ i=A,B,C,D,E$ such that
\[\tilde{s}=\arg\max_{s_A\in S_A}\min_{s_B\in S_B}u_A(s_A,s_B,\hat{s},\hat{s},\tilde{s})=\arg\min_{s_B\in S_B}\max_{s_A\in S_A}u_A(s_A,s_B,\hat{s},\hat{s},\tilde{s}),\]
\[\hat{s}=\arg\max_{s_C\in S_C}\min_{s_D\in S_D}u_C(\tilde{s},\tilde{s},s_C,s_D,\tilde{s})=\arg\min_{s_D\in S_D}\max_{s_C\in S_C}u_C(\tilde{s},\tilde{s},s_C,s_D,\tilde{s}),\]
\begin{align*}
&\max_{s_A\in S_A}\min_{s_B\in S_B}u_A(s_A,s_B,\hat{s},\hat{s},\tilde{s})=\min_{s_B\in S_B}u_A(\tilde{s},s_B,\hat{s},\hat{s},\tilde{s})=\min_{s_B\in S_B}\max_{s_A\in S_A}u_A(s_A,s_B,\hat{s},\hat{s},\tilde{s})\\
&=\max_{s_A\in S_A}u_A(s_A,\tilde{s},\hat{s},\hat{s},\tilde{s}),
\end{align*}
and
\begin{align*}
&\max_{s_C\in S_C}\min_{s_D\in S_D}u_C(\tilde{s},\tilde{s},s_C,s_D,\tilde{s})=\min_{s_D\in S_D}u_C(\tilde{s},\tilde{s},\hat{s},s_D,\tilde{s})=\min_{s_D\in S_D}\max_{s_C\in S_C}u_C(\tilde{s},\tilde{s},s_C,s_D,\tilde{s})\\
&=\max_{s_C\in S_C}u_C(\tilde{s},\tilde{s},s_C,\hat{s},\tilde{s}).
\end{align*}
Since
\[u_A(\tilde{s},s_B,\hat{s},\hat{s},\tilde{s})\leq \max_{s_A\in S_A}u_A(s_A,s_B,\hat{s},\hat{s},\tilde{s}),\]
\[\min_{s_B\in S_B}u_A(\tilde{s},s_B,\hat{s},\hat{s},\tilde{s})=\min_{s_B\in S_B}\max_{s_A\in S_A}u_A(s_A,s_B,\hat{s},\hat{s},\tilde{s}),\]
we get
\[\arg\min_{s_B\in S_B}u_A(\tilde{s},s_B,\hat{s},\hat{s},\tilde{s})=\arg\min_{s_B\in S_B}\max_{s_A\in S_A}u_A(s_A,s_B,\hat{s},\hat{s},\tilde{s})=\tilde{s}.\]
Similarly, from
\[u_C(\tilde{s},\tilde{s},\hat{s},s_D,\tilde{s})\leq \max_{s_C\in S_C}u_C(\tilde{s},\tilde{s},s_C,s_D,\tilde{s}),\]
\[\min_{s_D\in S_D}u_C(\tilde{s},\tilde{s},\hat{s},s_D,\tilde{s})=\min_{s_D\in S_D}\max_{s_C\in S_C}u_C(\tilde{s},\tilde{s},s_C,s_D,\tilde{s}),\]
we get
\[\arg\min_{s_D\in S_D}u_C(\tilde{s},\tilde{s},\hat{s},s_D,\tilde{s})=\arg\min_{s_D\in S_D}\max_{s_C\in S_C}u_C(\tilde{s},\tilde{s},s_C,s_D,\tilde{s})=\hat{s}.\]
Since
\[u_A(s_A,\tilde{s},\hat{s},\hat{s},\tilde{s})\geq \min_{s_B\in S_B}u_A(s_A,s_B,\hat{s},\hat{s},\tilde{s}),\]
and
\[\max_{s_A\in S_A}u_A(s_A,\tilde{s},\hat{s},\hat{s},\tilde{s})=\max_{s_A\in S_A}\min_{s_B\in S_B}u_A(s_A,s_B,\hat{s},\hat{s},\tilde{s}),\]
we obtain
\[\arg\max_{s_A\in S_A}u_A(s_A,\tilde{s},\hat{s},\hat{s},\tilde{s})=\arg\max_{s_A\in S_A}\min_{s_B\in S_B}u_A(s_A,s_B,\hat{s},\hat{s},\tilde{s})=\tilde{s}.\]
Similarly, from
\[u_C(\tilde{s},\tilde{s},s_C,\hat{s},\tilde{s})\geq \min_{s_D\in S_D}u_C(\tilde{s},\tilde{s},s_C,s_D,\tilde{s}),\]
and
\[\max_{s_C\in S_C}u_C(\tilde{s},\tilde{s},s_C,\hat{s},\tilde{s})=\max_{s_C\in S_C}\min_{s_D\in S_D}u_C(\tilde{s},\tilde{s},\hat{s},s_D,\tilde{s}),\]
we obtain
\[\arg\max_{s_C\in S_C}u_C(\tilde{s},\tilde{s},s_C,\hat{s},\tilde{s})=\arg\max_{s_C\in S_C}\min_{s_D\in S_D}u_C(\tilde{s},\tilde{s},s_C,s_D)=\hat{s}.\]
Therefore,
\[u_A(\tilde{s},s_B,\hat{s},\hat{s},\tilde{s})\geq u_A(\tilde{s},\tilde{s},\hat{s},\hat{s},\tilde{s})\geq u_A(s_A,\tilde{s},\hat{s},\hat{s},\tilde{s}),\]
\[u_C(\tilde{s},\tilde{s},\hat{s},s_D,\tilde{s})\geq u_C(\tilde{s},\tilde{s},\hat{s},\hat{s},\tilde{s})\geq u_C(\tilde{s},\tilde{s},s_C,\hat{s},\tilde{s}).\]
By symmetry we get
\[u_A(\tilde{s},\tilde{s},\hat{s},\hat{s},s_E)\geq u_A(\tilde{s},\tilde{s},\hat{s},\hat{s},\tilde{s})\geq u_A(s_A,\tilde{s},\hat{s},\hat{s},\tilde{s}),\]
\[u_B(\tilde{s},\tilde{s},\hat{s},\hat{s},s_E)\geq u_B(\tilde{s},\tilde{s},\hat{s},\hat{s},\tilde{s})\geq u_B(\tilde{s},s_B,\hat{s},\hat{s},\tilde{s}),\]
\[u_E(\tilde{s},s_B,\hat{s},\hat{s},\tilde{s})\geq u_E(\tilde{s},\tilde{s},\hat{s},\hat{s},\tilde{s})\geq u_E(\tilde{s},\tilde{s},\hat{s},\hat{s},s_E),\]
\[u_D(\tilde{s},\tilde{s},s_C,\hat{s},\tilde{s})\geq u_D(\tilde{s},\tilde{s},\hat{s},\hat{s},\tilde{s})\geq u_D(\tilde{s},\tilde{s},\hat{s},s_D,\tilde{s}),\]
and so on.

Thus, $(s_A,s_B,s_C,s_D,s_E)=(\tilde{s},\tilde{s},\hat{s},\hat{s},\tilde{s})$ is a Nash equilibrium which is symmetric in each group.
\end{proof}

\section{An example: Relative profit maximizing oligopoly in each group with two groups}\label{ex}

Consider a five-players game. The players are A, B, C, D and E. Suppose that the payoff functions of Players A, B and E are symmetric, and those of Players C and D are symmetric. The payoff functions of the players are
\begin{align*}
\pi_A=&(a-x_A-x_B-x_E-bx_C-bx_D)x_A-c_Ax_A-\frac{1}{2}[(a-x_A-x_B-x_E-bx_C-bx_D)x_B-c_Ax_B\\
&+(a-x_A-x_B-x_E-bx_C-bx_D)x_E-c_Ax_E],
\end{align*}
\begin{align*}
\pi_B=&(a-x_A-x_B-x_E-bx_C-bx_D)x_B-c_Ax_B-\frac{1}{2}[(a-x_A-x_B-x_E-bx_C-bx_D)x_A-c_Ax_A\\
&+(a-x_A-x_B-x_E-bx_C-bx_D)x_E-c_Ax_E],
\end{align*}
\begin{align*}
\pi_E=&(a-x_A-x_B-x_E-bx_C-bx_D)x_E-c_Ax_E-\frac{1}{2}[(a-x_A-x_B-x_E-bx_C-bx_D)x_A-c_Ax_A\\
&+(a-x_A-x_B-x_E-bx_C-bx_D)x_B-c_Ax_B],
\end{align*}
\[\pi_C=(a-x_C-x_D-bx_A-bx_B-bx_E)x_C-c_Cx_C-[(a-x_C-x_D-bx_A-bx_B-bx_E)x_D-c_Cx_D],\]
\[\pi_D=(a-x_C-x_D-bx_A-bx_B-bx_E)x_D-c_Cx_D-[(a-x_C-x_D-bx_A-bx_B-bx_E)x_C-c_Cx_C],\]
This is a model of relative profit maximization in each group in a five firms oligopoly with two groups. $x_A$, $x_B$, $x_C$, $x_D$ and $x_E$ are the outputs of the firms, and $p_A$, $p_B$, $p_C$, $p_D$ and $p_E$ are the prices of their goods. The demand functions are symmetric for Firms A, B and E, and they have the same cost functions. On the other hand, the demand functions are symmetric for Firms C and D, and they have the same cost functions. However, the demand function for Firm A (or B or E) is not symmetric for Firm C (or D), and the demand function for Firm C (or D) is not symmetric for Firm A (or B or E). Firm A's (or Firm B's or Firm E's) cost function is different from the cost function of Firm C (or Firm D). The cost functions of the firms are linear and there is no fixed cost.

We assume that Firm A (or B or E) maximizes its profit relatively to the profit of Firm B and E (or A and E, or A and B), and  Firm C (or D) maximizes its profit relatively to the profit of Firm D (or C). Note that
\[\pi_A+\pi_B+\pi_E=0,\ \pi_C+\pi_D=0.\]
Thus, this is a model of zero-sum game in each group with two groups.

Under the assumption of Cournot type behavior, the equilibrium outputs are
\[x_A=\frac{bc_C-c_A-ab+a}{3(1-b)(1+b)},\]
\[x_B=\frac{bc_C-c_A-ab+a}{3(1-b)(1+b)},\]
\[x_C=\frac{bc_A-c_C-ab+a}{2(1-b)(1+b)},\]
\[x_D=\frac{bc_A-c_C-ab+a}{2(1-b)(1+b)},\]
\[x_E=\frac{bc_A-c_C-ab+a}{3(1-b)(1+b)}.\]
The equilibrium prices of the goods are
\[p_A=c_A,\]
\[p_B=c_A,\]
\[p_C=c_C,\]
\[p_D=c_C,\]
\[p_E=c_A.\]
Therefore, the prices of the goods are equal to the marginal costs in each group.

The maximin and minimax strategies between Firms A and B are
\[\arg\max_{x_A}\min_{x_B}\pi_A,\ \arg\min_{x_B}\max_{x_A}\pi_A.\]
Those between Firm C and D are
\[\arg\max_{x_C}\min_{x_D}\pi_C,\ \arg\min_{x_D}\max_{x_C}\pi_C.\]
Those between Firms A and E and so on are similarly defined.

In our example we obtain
\[\arg\max_{x_A}\min_{x_B}\pi_A=\frac{bc_C-c_A-ab+a}{3(1-b)(1+b)},\]
\[\arg\min_{x_B}\max_{x_A}\pi_A=\frac{bc_C-c_A-ab+a}{3(1-b)(1+b)},\]
\[\arg\max_{x_C}\min_{x_D}\pi_C=\frac{bc_A-c_C-ab+a}{2(1-b)(1+b)},\]
\[\arg\min_{x_D}\max_{x_C}\pi_C=\frac{bc_A-c_C-ab+a}{2(1-b)(1+b)},\]
\[\arg\max_{x_A}\min_{x_E}\pi_A=\frac{bc_C-c_A-ab+a}{3(1-b)(1+b)},\]
and so on. They are the same as Nash equilibrium strategies.

\section{Concluding Remark}

In this paper we have examined the relation between Sion's minimax theorem for a continuous function and a Nash equilibrium in . We want to extend this result to more general multi-players game.

\end{document}